\documentclass[a4paper]{article}
\usepackage[pdftex]{hyperref}
\usepackage{amsmath}
\usepackage[english]{babel}
\usepackage{latexsym}
\usepackage{amssymb}
\usepackage{amscd}
\usepackage{amsgen,amstext,amsbsy,amsopn}
\usepackage{math rsfs}

\usepackage{amsthm,epsfig,graphicx,graphics}
\usepackage[latin1]{inputenc}
\usepackage{xspace}
\usepackage{amsxtra}

\usepackage{color}

\usepackage{amstext}
\usepackage{amsxtra}

\usepackage{latexsym}

\usepackage{amscd}

\usepackage{amsgen,amstext,amsbsy,amsopn}
\usepackage{math rsfs}

\numberwithin{equation}{section}

\newcommand{\eps}{\varepsilon}

\newcommand{\gpf}{\mathcal{E}^{\mathrm{GP}}}

\newcommand{\lllf}{\mathcal{E}^{\mathrm{LLL}}}

\newcommand{\lllff}{\mathcal{F}^{\mathrm{LLL}}}
\newcommand{\lllee}{F^{\mathrm{LLL}}}

\newcommand{\tff}{\mathcal{E}^{\mathrm{TF}}}

\newcommand{\tfe}{E^{\mathrm{TF}}}
\newcommand{\tfm}{\rho^{\mathrm{TF}}}

\newcommand{\tfchem}{\mu^{\mathrm{TF}}}

\newcommand{\Z}{\mathbb{Z}}
\newcommand{\R}{\mathbb{R}}
\newcommand{\N}{\mathbb{N}}

\newcommand{\C}{\mathbb{C}}

\newcommand{\F}{\mathcal{F}}
\newcommand{\E}{\mathcal{E}}

\newcommand{\al}{\alpha}

\newcommand{\ep}{\varepsilon}

\newcommand{\Om}{\Omega}

\newcommand{\dd}{\partial}

\newcommand{\Fo}{\mathfrak{F}_{\ep}}
\newcommand{\Fos}{\mathfrak{F}_{\ep} ^s}

\newcommand{\ut}{u_{\tau}}
\newcommand{\ft}{f_{\tau}}

\newcommand{\fat}{f_{\al,\tau}}

\newcommand{\Pie}{\Pi_{\ep}}

\def\Xint#1{\mathchoice
   {\XXint\displaystyle\textstyle{#1}}%
   {\XXint\textstyle\scriptstyle{#1}}%
   {\XXint\scriptstyle\scriptscriptstyle{#1}}%
   {\XXint\scriptscriptstyle\scriptscriptstyle{#1}}%
   \!\int}
\def\XXint#1#2#3{{\setbox0=\hbox{$#1{#2#3}{\int}$}
     \vcenter{\hbox{$#2#3$}}\kern-.5\wd0}}

\newtheorem{teo}{Theorem}[section]
\newtheorem{lem}{Lemma}[section]

\newcounter{remark}[section]
\begin{document}

%\markboth{\scriptsize{Vortex Rings in rotating BECs}}{\scriptsize{Vortex Rings in rotating BECs}}

\title{Annular Bose-Einstein Condensates in the Lowest Landau Level}

\author{N. Rougerie	\\	\normalsize\it	Universit\'{e} Paris 6 \\ \normalsize\it Laboratoire Jacques-Louis Lions \\ \normalsize\it 4 place Jussieu, 75006 Paris, France \\ \normalsize\it \hspace{-.5 cm}}

\date{November 17, 2010}

\maketitle

\begin{abstract}
A rotating superfluid such as a Bose-Einstein condensate is usually described by the Gross-Pitaevskii (GP) model. An important issue is to determine from this model the properties of the quantized vortices that a superfluid nucleates when set into rotation. In this paper we address the minimization of a two dimensional GP energy functional describing a rotating annular Bose-Einstein condensate. In a certain limit it is physically relevant to restrict the minimization to the Lowest-Landau-Level, that is the first eigenspace of the Ginzburg-Landau operator. Taking the particular structure of this space into account we obtain theoretical results concerning the vortices of the condensate. We also compute the vortices' locations by a numerical minimization procedure. We find that they lie on a distorted lattice and that multiple quantized vortices appear in the central hole of low matter density.  \\

\vspace{0,2cm}
\vspace{0,2cm}
	
\end{abstract}

%\begin{abstract} 
%	\vspace{0,2cm}

%\end{abstract}

\tableofcontents

\section{Introduction}\label{sec:Intro}

A Bose-Einstein condensate (BEC) is an object exhibiting quantum properties on a macroscopic scale. The effects of such phenomena as superfluidity can be observed experimentally on these systems, in particular the nucleation of quantized vortices in a rotating BEC (see for example the review \cite{Fe2} or \cite{MAH,MCW,ARV,RAV,HCE,ECH} for reports on original experiments). The most common description for a BEC is through a macroscopic wave function $\psi : \R^d \mapsto \C$ which modulus squared will give the matter density profile. Here $d=2,3$ and in this paper we will focus on a situation where it is physically justified to consider a two-dimensional BEC. For a condensate at equilibrium in the rotating frame, this wave function should minimize the following \emph{Gross-Pitaevskii functional}
\begin{equation}\label{EGP-} 
\gpf[\psi]=\int_{\R^2} \left( \frac{1}{2}\left| \nabla \psi -i\Omega x^{\perp}\psi \right|^2 +\left( V(x) - \frac{\Omega^2}{2} \vert x\vert ^2 \right) \vert \psi \vert^2 +\frac{G}{2}\vert \psi \vert ^4 \right) dx,
\end{equation}
under the mass constraint 
\begin{equation}\label{masseGP-}
\int_{\R ^2} \vert \psi \vert ^2=1.
\end{equation}
Here $x=(x_1,x_2)\in \R^2$, $x^{\perp}=(-x_2,x_1)$, $\Om$ is the angular velocity at which the condensate is rotated around the axis perpendicular to the plane $x_1,x_2$ and $G$ is a dimensionless coefficient characterizing the strength of atomic interactions. An essential feature of the experimental set-ups is the trapping potential, denoted by $V$ in (\ref{EGP-}). It is this potential that confines the atoms in a bounded region and prevents them from flying apart under the action of the centrifugal forces. For most experiments it takes the form (in the appropriate units that we have implicitly chosen in (\ref{EGP-}))
\begin{equation}\label{harmpot}
V(x) =  \frac{1}{2}|x| ^2.
\end{equation} 
Such a potential sets a limit to the rotational speed $\Om$ that one can impose to the condensate. Indeed, one can see that with $V$ chosen as above the second term in (\ref{EGP-}) is not bounded below if $\Om > 1$. Consequently the minimization problem makes no sense. Physically this corresponds to the fact that the centrifugal force overcomes the trapping force and drives the condensate out of the trap.\\
To prevent such a singular behavior in the limit $\Om \to 1$, Fetter \cite{Fe2} proposed to use instead a potential of the form 
\begin{equation}\label{quartpot}
V(x) = \frac{1}{2}|x| ^2 + \frac{k}{4}|x| ^4
\end{equation}
which, at least theoretically, allows for arbitrarily large rotation speeds $\Om$. Such a potential is a good approximation for those that have been used in the experiments (\cite{Exp1,Exp2}) and this is the one we shall consider in this paper. \\

An important experimental test for the superfluidity of a BEC is the observation of \emph{quantized vortices}, which are described in Gross-Pitaevskii theory as zeros of the wave function carrying a positive \emph{topological degree} or winding number (corresponding to the quantized phase circulation around the vortex). When a BEC is rotated, it will contain more and more vortices as the rotation speed is increased, until they become densely packed in the condensate when $\Om$ approaches $1$. Vortices repel each other, and to minimize this effect they will arrange on regular hexagonal lattices when $\Om \to 1$, for both the harmonic trap (\ref{harmpot}) and the `anharmonic' one (\ref{quartpot}). These Abrikosov lattices (named by analogy with the physics of type II superconductors, \cite{Abri}) have been observed experimentally, see e.g. \cite{ARV} for experiments with the trap (\ref{harmpot}) and  \cite{Exp1} for experiments with the trap (\ref{quartpot}). It should be noted that all vortices in these lattices are simply quantized, i.e. they have degree $1$.\\
A striking new feature of the trap (\ref{quartpot}) is the possibility to create annular condensates by increasing the rotation speed beyond the barrier that would be imposed by a weaker confinement. Indeed, for sufficiently large rotation speeds, the centrifugal forces will dip a hole in the center of the condensate. Interesting questions then arise about the vortex structure of the condensate. For example : does the hexagonal vortex lattice survive the formation of the hole ? When the hole has just formed the answer is yes, see \cite{FJS,FB,KB} for references in the physics literature and \cite{CY} for a mathematically rigorous treatment. At very large rotation speeds however, the vortices disappear from the bulk of the condensate, resulting in a so-called giant vortex phase, a situation that we do not address here and for which we refer to \cite{FJS,FB,KB,KF} for physical references and \cite{CRY,R} for a mathematical analysis.\\
The locations of vortices in the central hole of low matter density can also be investigated. It is commonly known (see e.g. \cite{FJS,KB,KTU}) that there should be a phase circulation around this hole, indicating that vortices lie within it. In regions of high matter density, it is most favorable for the vortices to have degree $1$, but it is far from obvious that it should also be the case in regions of low matter density. In particular, a question arising naturally \cite{ABD,FJS} is wether the vortices in the central hole of an annular Bose-Einstein condensate gather in a single multiply quantized vortex or form a pattern of simply quantized ones.\\

In this paper we address the case of a two-dimensional annular Bose-Einstein condensate rotating in a potential of the form (\ref{quartpot}). We will study the formation of the annular condensate by taking the limit $k\to 0$ of the problem, with the rotation speed $\Om$ properly scaled so as to capture the regime where the central hole is created under the influence of the centrifugal force. In such a parameter regime it is physically justified to consider a simplified variational problem that we describe in the next section. The very particular structure of the variational space considered allows to obtain interesting theoretical results. It also suggests a very natural numerical minimization method that we have used to confirm and extend our analytical results. The results we are going to present first appeared in a note intended for physicists \cite{Moi} to which we refer for further physical discussion (see also \cite{AM}).

\subsection{The GP energy and the LLL reduction}\label{sousec:setting}

We consider the GP energy functional 
\begin{equation}\label{EGP} 
\gpf[\psi]=\int_{\R^2} \left( \frac{1}{2}\left| \nabla \psi -i\Omega x^{\perp}\psi \right|^2 +\left( \frac{1-\Omega^2}{2} \vert x\vert ^2 +\frac{k}{4}\vert x\vert ^4 \right) \vert \psi \vert^2 +\frac{G}{2}\vert \psi \vert ^4 \right) dx,
\end{equation}
to be minimized under the mass constraint 
\begin{equation}\label{masseGP}
\int_{\R ^2} \vert \psi \vert ^2=1.
\end{equation}
The main idea of our analysis is to restrict the minimization of $\gpf$ to the first eigenspace of the Ginzburg-Landau operator $-\left(\nabla -i\Omega x^{\perp}\right)^2$, corresponding to the eigenvalue $\Omega$. This is the Lowest Landau Level (LLL), introduced in the context of Bose-Einstein condensation by Ho \cite{Ho}. \\
The use of such a simplification has originally been introduced for condensates rotating in traps of the form (\ref{harmpot}), i.e. for the functional 
\begin{equation}\label{quad quart 1: EGP quad}
\gpf[\psi]=\int_{\R^2} \left( \frac{1}{2}\left| \nabla \psi -i\Omega x^{\perp}\psi \right|^2 +\frac{1-\Omega^2}{2} \vert x\vert ^2  \vert \psi \vert^2 +\frac{G}{2}\vert \psi \vert ^4 \right) dx.
\end{equation}
Restricting the minimization to the LLL is then justified in the limit $\Om \to 1$. Indeed, the first term in the energy is at least equal to $\Om \sim 1$ since $\Om$ is the first eigenvalue of $-\left(\nabla -i\Omega x^{\perp}\right)^2$. On the other hand, the second term will obviously be small in the same limit. Thus, minimizing the sum of the second and the third term one also obtains something much smaller than $\Om$.\\
This indicates that the energy must stay close to $\Om$ when $\Om \rightarrow 1$. But the spectral gap between the first and the second eigenvalue of the Ginzburg-Landau operator is equal to $\Om$. To obtain an energy of order $\Om$ the projection of a minimizer of the functional (\ref{quad quart 1: EGP quad}) on the excited energy levels of $-\left(\nabla -i\Omega x^{\perp}\right)^2$ must then be negligible.\\
This formal justification has been the basis of several works on harmonically trapped BECs in the regime $\Om \to 1$, e.g. \cite{ABD,CKR,WBP}. A rigorous justification has been given in \cite{AB-2D}. For other mathematical developments on the functional (\ref{quad quart 1: EGP quad}) in the regime $\Om \rightarrow 1$ (Landau regime), we refer to \cite{AB,ABNmath}.\\

The LLL consists (see e.g \cite{LP}) of functions of the form $\psi(z)=f(z)e^{-\Omega \vert z \vert^2 /2}$ where $z$ is the complex variable $x_1+ix_2$ and $f$ is a holomorphic function. For an LLL function $\psi$, the Gross-Pitaevskii energy (\ref{EGP}) reduces to (we denote $dz = dx_1 dx_2$)
\[
\gpf[\psi] = \Omega + \int_{\C} \left( \left(\frac{1-\Omega^2}{2}\vert z \vert ^2  + \frac{k}{4} \vert z \vert ^4 \right)\vert \psi \vert ^2 + \frac{G}{2} \vert \psi \vert ^4\right)dz.
\]
In this paper we will rather consider functions of the form $\psi(z)=f(z)e^{-\vert z \vert^2 /2}$ (thus we effectively restrict the minimization to the first eigenspace of $-\left(\nabla -i x^{\perp}\right)^2$) to have a functional space independent of $\Om$. This approximation is harmless because, as we will see, our analysis is concerned with rotation speeds $\Om$ close to $1$. As for the energy, the difference is very small as showed by \cite[Lemma 3.1]{AB-2D} :  with $\psi$ of the form $\psi(z)=f(z)e^{-\vert z \vert^2 /2}$ with $f$ holomorphic we have
\begin{eqnarray}
\gpf[\psi] &=& \Omega +\lllf[\psi] \label{decompenergie} \\
\lllf[\psi]&=& \int_{\C} \left( \left( \left(1-\Omega \right)\vert z \vert ^2  + \frac{k}{4} \vert z \vert ^4 \right)\vert \psi \vert ^2 + \frac{G}{2} \vert \psi \vert ^4\right)dz, \label{ELLL}
\end{eqnarray}
so that we will minimize the energy $\lllf$ amongst functions of the form $\psi(z)=f(z)e^{-\vert z \vert^2 /2}$ with $f$ holomorphic, under the mass constraint $\int \vert \psi \vert ^2=1$. For this purpose we will use the mathematical framework introduced in \cite{ABNphi,ABNmath} for the study of the functional (\ref{quad quart 1: EGP quad}). See also \cite{Nie2} for the corresponding dynamical model.\\
Note that the validity of such a reduction could be investigated with the tools of \cite{AB-2D} (as well as the validity of the reduction to a two-dimensional model), but in this paper we shall be concerned only with the study of the reduced energy functional $\lllf$. We will provide in Subsection \ref{sousec:discu} an a posteriori criterion for the validity of the reduction. \\

We now describe the parameter regime that we shall consider. Let us define a small parameter 
\begin{equation}\label{epsilon}
 \varepsilon=k^{1/3}
\end{equation}
corresponding to a small anharmonicity regime and study the asymptotics of the problem as $\varepsilon \rightarrow 0$. We take $\Om$ satisfying
\begin{equation}\label{Omega}
\Omega = 1 + \beta k^{2/3} 
\end{equation}
and will consider $\beta$ and $G$ as fixed. This choice will lead to a functional with bounded coefficients (see (\ref{energieLLL}) below). In this regime the radius of the condensate is not bounded as a function of $\ep$. We thus rescale distances by making the change of variables 
\begin{equation}\label{scaling}
\phi(z)=\varepsilon ^{-1/2}\psi(\ep ^{-1/2}z)
\end{equation}
and for every $\phi$ we define the function 
\begin{equation}\label{chgtfunction}
f(z)=\phi(z)e^{\vert z \vert^2/2\varepsilon}.
\end{equation}
By definition of the LLL, $f$ belongs to the Fock-Bargmann space \cite{Bar}
\begin{equation}\label{Fock}
\Fo := \left\lbrace f \mbox{ holomorphic }, \int_{\C} \vert f\vert^2 e^{-\vert z \vert^2/\varepsilon}dz < \infty \right\rbrace. 
\end{equation}
The space $\Fo$ is a Hilbert space for the scalar product 
\begin{equation}\label{scalaireFoc} 
\left\langle f,g \right\rangle := \int_{\C} \overline{f(z)}g(z) e^{-\vert z \vert^2/\varepsilon}dz.
\end{equation}
The point of introducing such a space is that the orthogonal projector from $L^2 \left(\C, e^{-|z|^2/ \ep} dz \right)$ onto $\Fo$ is explicitly known \cite{Mar,Fol} and called the Szegö projector:
\begin{equation}\label{projecteur}
\Pi_{\varepsilon}(g)(z)=\frac{1}{\pi \varepsilon}\int_{\C}e^{z\bar{z'}/\varepsilon}e^{-\vert z' \vert^2 /\varepsilon} g(z')dz'.
\end{equation}
We will also use the spaces
\begin{equation}\label{Focks}
\Fos := \left\lbrace f \mbox{ holomorphic }, \int_{\C} \left(1+|z| ^2 \right) ^{s} \vert f \vert^2 e^{-\vert z \vert^2/\varepsilon}dz < \infty \right\rbrace 
\end{equation}
that we equip with the norms
\begin{equation}\label{normeFocks}
\Vert f \Vert_{\Fos} : = \left(\int_{\C} \left(1+|z| ^2 \right) ^{s} \vert f \vert^2 e^{-\vert z \vert^2/\varepsilon}dz \right) ^{1/2}.
\end{equation}

After the change of scale (\ref{scaling}) and the change of function (\ref{chgtfunction}), the energy becomes 
\begin{equation}\label{scaleenergie}
\lllf [\psi] = \varepsilon \lllff [f]
\end{equation}
with
\begin{equation}\label{energieLLL}
\lllff [f] =\int_{\C} \left(\left( -\beta \vert z \vert ^2  + \frac{1}{4} \vert z \vert ^4 \right)\vert f \vert ^2 e^{-\vert z \vert^2 /\varepsilon} + \frac{G}{2} \vert f \vert ^4 e^{-2\vert z \vert^2 /\varepsilon}\right)dz
\end{equation}
and we thus have to minimize (\ref{energieLLL}) under the mass constraint
\begin{equation}\label{masseLLL}
\Vert f \Vert_{\Fo}=1.
\end{equation}
It is this minimization problem that we will analyze in the sequel, namely we look at the problem
\begin{equation}\label{miniLLL}
\lllee := \inf \left\{ \lllff[f] \: \vert \: f\in \Fo ^1,\:  \Vert f \Vert_{\Fo} = 1 \right\}.
\end{equation}
In particular we will see that this problem admits a solution, i.e. the infimum above is actually a minimum. We are able, using the Szegö projector (\ref{projecteur}) to derive the Euler-Lagrange equation satisfied by any solution of the minimization problem. Using this equation and the very particular structure of the Fock-Bargmann space (in particular the constraint that $f$ has to be analytic) we are able to derive that any minimizer must have an infinite number of zeros (vortices) if $\eps$ is small enough. We also construct almost critical points, functions that are solutions to the Euler-Lagrange equation up to a small remainder term. Evaluating the energy of these almost critical points yields an upper bound to the energy that we believe is optimal although we are not able to prove it. This upper bound will be useful to discuss the validity of the LLL reduction in Subsection \ref{sousec:discu} below.

\subsection{Main Analytical Results}\label{sousec:results}

Let us start by showing that the problem we want to consider is actually well-posed :

\begin{teo}[\textbf{Well-posedness and Euler-Lagrange equation}]\mbox{}\\ \label{theo:minimisation}
For any $\ep > 0$ fixed, the problem (\ref{miniLLL}) admits a solution in $\Fo ^2$. Any minimizer is a solution to
\begin{equation}\label{EEL1}
-\beta M_{\varepsilon} f+\frac{1}{4}\left( M_{\varepsilon}^2 f + \varepsilon M_{\varepsilon} f\right)+G\Pi_{\varepsilon}(e^{-\vert z \vert^2 /\varepsilon}\vert f\vert^2 f)=\mu f 
\end{equation}
where $\mu$ is the Lagrange multiplier coming from the mass constraint and $M_{\varepsilon}$ is the operator defined by
\begin{equation}\label{defiMep}
 M_{\varepsilon} f = \varepsilon \partial_z \left( z f \right) .
\end{equation}
Alternatively, the Euler-Lagrange equation (\ref{EEL1}) may be written
\begin{equation}\label{EEL3}
(-\beta+\frac{1}{4}\varepsilon)\Pi_{\varepsilon}(\vert z \vert^2 f)+\frac{1}{4}\Pi_{\varepsilon}\left( \vert z\vert^2 \Pi_{\varepsilon}(\vert z\vert^2 f)\right)+G\Pi_{\varepsilon}(^{-\vert z \vert^2 /\varepsilon}\vert f\vert^2 f)=\mu f
\end{equation}
or
\begin{equation}\label{EEL2}
-\beta M_{\varepsilon} f+\frac{1}{4}\left( M_{\varepsilon}^2 f + \varepsilon M_{\varepsilon} f\right)+\frac{G}{2}\bar{f}(\varepsilon \partial_{z})[f^2(z/2)]=\mu f, 
\end{equation}
where the operator $\bar{f}(\varepsilon \partial_{z})$ is defined as
\[
 \bar{f}(\varepsilon \partial_{z})[g] := \sum ^{+\infty}_{k=0}\overline{a_k}(\ep \partial _z)^k g
\]
if $f(z)=\sum_{k=0} ^{+\infty} a_k z^k$.
\end{teo}

Next we show that a minimizer of (\ref{miniLLL}) cannot have a finite number of zeros if $\ep$ is small enough. In particular, the minimum cannot be achieved by a polynomial.

\begin{teo}[\textbf{Infinite number of zeros}]\mbox{}\\ \label{theo:infivortex}
If $\ep$ is sufficiently small, any solution to the minimization problem (\ref{miniLLL}) has an infinite number of zeros.
\end{teo}

We will discuss this result in further details in Subsection \ref{sousec:discu} and now describe the procedure we follow to construct approximate solutions to the Euler-Lagrange equation (\ref{EEL3}).\\

As is the case for harmonically trapped condensates in fast rotation, we expect that in the range of parameters we explore the scales of the problem will decouple. Namely we expect any minimizer of (\ref{energieLLL}) to be of the form $\alpha u$ where $u$ varies on the scale of the vortex pattern (which is small compared to the size of the condensate) and $\alpha$ is a slow varying profile giving the general shape of the condensate. The kind of profile $\alpha$ we have in mind is a Thomas-Fermi distribution. More precisely, $\al$ should look like a minimizer of the energy (\ref{energieLLL}) without the holomorphy constraint. Such a function has compact support and can thus not be in the LLL, but one can approach it by an LLL function which is an almost critical point for the energy (\ref{energieLLL}). We proceed as follow. Let us introduce 
\begin{equation}\label{Abriko}
  u _{\tau}(z)=e^{-{|z|^{2}/2\ep}}f _{\tau}(z),\quad
  f _{\tau}(z)=e^{ z^{2}/2\ep}\Theta\left(\sqrt{\frac{\tau_{I}}{\pi \ep}}z, \tau\right)
\end{equation}
where $\tau=\tau_R+i\tau_I$ is any complex number and 
\begin{equation}\label{eq:theta}
\Theta(v,\tau)=\frac{1}{i}\sum_{n=-\infty}^{+\infty}(-1)^{n}e^{i\pi\tau(n+1/2)^{2}}
e^{(2n+1)\pi iv}
\end{equation} 
is the Jacobi $\Theta$ function (more precisely the $\Theta_0$ function according to the classification of \cite{Nie1}). It has the property (see \cite{Cha} for more details) 
\begin{equation}\label{propTheta1}
\Theta(v+k+l\tau,\tau)=(-1)^{k+l}e^{-2i\pi lv} e^{-i\pi l\tau}\Theta(v,\tau).
\end{equation}
The interest of introducing such functions is twofold. Firstly it is known \cite{Cha} that any function $v$ whose modulus is periodic over the lattice $\sqrt{\frac{\pi\ep}{\tau_I}}\Z\oplus \sqrt{\frac{\pi\ep}{\tau_I}}\Z\tau$, vanishes exactly on the points of the lattice with simple zeros and such that $g=ve^{\vert z \vert^2/2\ep }$ is holomorphic must be proportional to $u_{\tau}$. Secondly, the function $f_{\tau}$ is a solution to the Abrikosov problem (see e.g. \cite{Abri} and \cite[Section 4]{ABNmath} for a detailed discussion)
\begin{equation}\label{fabri}
\Pie (|f_{\tau} |^2e^{-|z|^2/\ep}f_{\tau} )=\lambda_{\tau} f_{\tau} ,\mbox{ with } \lambda_{\tau}=\Xint - | u _{\tau} |^2 b(\tau),
\end{equation} 
and 
\begin{equation}\label{btau}
b(\tau)=\frac{\Xint - |u _{\tau}|^4}{\left( \Xint - |u _{\tau} |^2 \right) ^2} = \sum_{k,l\in \Z} e^{-\pi | k \tau -l|^2/\tau_I}.
\end{equation}
Here we denote $\Xint - v$ the average of a periodic function $v$. Equation (\ref{fabri}) is similar to (\ref{EEL3}) without the potential term and with $\mu=\lambda_{\tau}$ so that one can expect to obtain a solution of (\ref{EEL3}) by a slight modification of $f_{\tau}$. We refer to \cite{ABNphi,ABNmath} and references therein for details on the quantity $b(\tau)$. Let us just mention that it is minimum ($b(\tau)\sim 1.16$) for $\tau = e^{2i\pi / 3}$, which corresponds to a hexagonal lattice. From now on we take 
\[
\tau = e^{2i\pi / 3}
\]
and denote $b= b(e^{2i\pi / 3})$.\\
To obtain our approximate solutions we first multiply $\ft$ by a slow varying profile $\al$ with compact support and $\int_{\R ^2} \vert \alpha \vert ^2 =1$. This trial state would correspond to having a regular hexagonal lattice of vortices in the bulk of the condensate (that is, the support of $\al$). To obtain an admissible test function for the problem (\ref{miniLLL}) we project $\al \ft$ onto the Fock-Bargmann space and normalize it :
\begin{equation}\label{deff}
f_{\alpha,\tau}=\frac{\Pi_{\ep}(\alpha f_{\tau})}{\left\Vert \Pi_{\ep}(\alpha f_{\tau}) \right\Vert_{\Fo}}.  
\end{equation}
Accordingly
\begin{equation}\label{defu}
 u_{\alpha,\tau}=f_{\alpha,\tau}e^{-\vert z \vert^2 / 2\ep}.
\end{equation}
As for the choice of the appropriate $\al$, it turns out to be related to the following Thomas-Fermi energy functional :
\begin{equation}\label{TFfunctionb}
\tff_b [\rho] = \int_{\C} \left( \left( -\beta \vert z \vert ^2  + \frac{1}{4} \vert z \vert ^4 \right) \rho   + \frac{bG}{2}  \rho  ^2 \right)dz.
\end{equation}
The associated minimization problem is
\begin{equation}\label{Tfenergyb}
\tfe_b := \inf \left\{ \tff_b [\rho] \: , \: \rho \geq 0 ,\: \int_{\R ^2} |z| ^2 \rho < +\infty, \:  \int_{\R ^2}  \rho  = 1 \right\}.
\end{equation}
Its solution is unique and given as
\begin{equation}\label{tfdensity}
\tfm_b : = \max\left( \frac{\tfchem_b +\beta |z|^2-\frac{1}{4}|z|^4}{bG}, 0\right)
\end{equation}
where the chemical potential $\tfchem_b$ appears as a Lagrange multiplier in the equation associated to the minimization (\ref{Tfenergyb}) ans is thus fixed by the constraint $\int_{\R ^2}  \tfm_b  = 1$. Note that there exists a critical value $\beta_c$ such that the support of such a function becomes an annulus when $\beta > \beta_c$ (see the discussion in Section \ref{sousec:discu}). This is in contrast with the situation where $k=0$ (purely quadratic trap) where the support of the TF minimizer is always a disc.\\\\
We have the following result :

\begin{teo}[\textbf{Almost critical points}]\label{theo:pointscri}\mbox{}\\
Let $\alpha$ be a profile satisfying
\begin{equation}\label{tfprofile}
 \left| \alpha (z) \right|^2 = \tfm_b
\end{equation}
and $\fat$ be the associated function via formula (\ref{deff}). There holds
\begin{equation}\label{resultEq}
(-\beta+\frac{1}{4}\varepsilon)\Pi_{\varepsilon}(\vert z \vert^2 f_{\alpha,\tau})+\frac{1}{4}\Pi_{\varepsilon}\left( \vert z\vert^2 \Pi_{\varepsilon}(\vert z\vert^2 f_{\alpha,\tau})\right) + G\Pi_{\varepsilon}(e^{-\vert z \vert^2 /\varepsilon}\vert f_{\alpha,\tau}\vert^2 f_{\alpha,\tau})=
\tfchem_b f_{\alpha,\tau}+R_{\ep}
\end{equation}
where $\Vert R_{\ep} \Vert_{\Fo} \leq C \ep ^{1/4}$. Moreover
\begin{equation}\label{resultEn}
\lllff [f_{\alpha,\tau}] = \int_{\C} \left( \left( -\beta \vert z \vert ^2  + \frac{1}{4} \vert z \vert ^4 \right)\vert \al \vert^2  + \frac{bG}{2} \vert \al \vert ^4 \right)dz +O(\ep ^{1/4})
%&=& \frac{3}{5}\left(\frac{3bG}{8\pi}\right)^{2/3}-\beta^{2} + O(\ep ^{1/4}).
\end{equation}
and thus
\begin{equation}\label{upperbound}
\lllee \leq \lllff [\fat] = \tfe _ b + O(\ep ^{1/4}).
\end{equation}

\end{teo}

\subsection{Discussion}\label{sousec:discu}

In this subsection we discuss some of the physical insights that one can obtain from our theorems, and some of the questions that can not be answered analytically, for which we will rely on numerical simulations.\\

Let us first comment a little bit more the ideas behind the construction leading to Theorem \ref{theo:pointscri}. It is known from the experiments that in the regime we consider, the condensate should contain a very large number of vortices. Theorem \ref{theo:infivortex} also suggests to use a trial function containing as many vortices as possible, although it does not give information on their locations.  On the other hand the trapping potential will force the condensate to live essentially in a bounded region, the complement of which will be a region of very low matter density. It is natural to think that, at least in the bulk, the vortices will form a regular hexagonal lattice. This is widely observed in experiments and numerical simulations.\\
Combining these two requirements, we arrive at a trial function of the form $c \al \ft$ where $\ft$ is defined in (\ref{Abriko}), $\al$ is a function with compact support and $c$ a normalization factor. Such a function is of course not in the Fock-Bargmann space : it has compact support and thus cannot be holomorphic. To obtain an admissible trial state we project it onto $\Fo$ using the Szegö projector. Computing the energy of such a trial function consists of essentially two steps (see Subsection \ref{sousec:pointscri} for the detailed proof). First, using a remarkable result from \cite{ABNmath} we have (see Lemma \ref{lem:ABN} below for a precise statement)
\[
\Pie (\al \ft) \sim \al \ft
\] 
when $\ep \rightarrow 0$. Thus the projection onto the Fock-Bargmann space has no effect on the energy to leading order and we obtain
\begin{multline*}
\lllff [f_{\alpha,\tau}] \sim\Vert \Pie (\al \ft) \Vert_{\Fo} ^{-2}  \int_{\C} 
 \left( -\beta \vert z \vert ^2  + \frac{1}{4} \vert z \vert ^4 \right)\vert \al \vert^2 |\ft| ^2 e^{-|z|^2/\ep} dz\\
+ \Vert \Pie (\al \ft) \Vert_{\Fo} ^{-4} \frac{G}{2} \int_{\C}  \vert \al \vert ^4 |\ft| ^4 e^{-2 |z|^2/\ep} dz.
\end{multline*}
Now, $|\ft|e^{-|z|^2 / 2 \ep} = |\ut|$ is periodic over a lattice whose period is very small ($\propto \ep ^{1/2}$), whereas $\al$ is chosen with a support of fixed size. Therefore one should expect a homogenization effect leading to 
\[
\lllff [f_{\alpha,\tau}] \sim  \Vert \Pie (\al \ft) \Vert_{\Fo} ^{-2} \Xint - |\ut| ^2 \int_{\C}  \left( -\beta \vert z \vert ^2  + \frac{1}{4} \vert z \vert ^4 \right)\vert \al \vert ^2  dz + \Vert \Pie (\al \ft) \Vert_{\Fo} ^{-4} \Xint - |\ut| ^4 \frac{G}{2} \int_{\C}  \vert \al \vert ^4dz.
\]
Computing the norm of $\Pie (\al \ft)$ uses the same ideas and leads to
\[
\Vert \Pie (\al \ft) \Vert_{\Fo} \sim \left( \Xint - |\ut| ^2 \right)^{1/2}.
\]
Finally, using (\ref{btau}),
\[
\lllff [f_{\alpha,\tau}] \sim \int_{\C} \left( \left( -\beta \vert z \vert ^2  + \frac{1}{4} \vert z \vert ^4 \right)\vert \al \vert ^2  +  \frac{bG}{2}  \vert \al \vert ^4\right)dz
\]
and one immediately sees that the optimal choice for $\al$ is the one we made in (\ref{tfprofile}). One should note that the only contribution of the vortex structure to the leading order of the energy is through the coefficient $b$, which depends on the type of lattice we choose and is minimum for the hexagonal lattice.\\
We do believe that this construction is optimal, but a proof of a lower bound matching (\ref{upperbound}) would probably require to show that the vortices of a true minimizer of the energy are located close to the sites of a regular lattice. This is a difficult question, linked to the crystallization problem, and remains a challenging open problem in more than one context, see \cite{ABNmath,SS} and references therein.\\
We have however the easy but non-optimal lower bound
\begin{equation}\label{lowerbound}
\tfe_1 \leq \lllee
\end{equation}
where $\tfe_1$ is defined by taking $b=1$ in (\ref{Tfenergyb}). This follows by minimizing $\lllff[f]$ with respect to $f e^{-|z|^2 / 2\ep}$, dropping the constraint that this function should be in the LLL. As $b$ is a constant, this lower bound confirms that $\tfe_b$ gives at least the order of magnitude of the energy. The gap between the upper and the lower bound lies in the coefficient $b$ that takes into account the density modification due to the presence of a large number of vortices. Filling this gap would be a first step towards a rigorous proof of the onset of the Abrikosov lattice in this regime.\\

Although we have no rigorous proof that the trial state (\ref{deff}) accurately describes the true state of affairs, it is useful for a physical discussion to give some details on this function. The analysis of a profile such as (\ref{tfdensity}) has already been carried out in \cite{FJS,ABD}, so we only adapt and summarize their results.\\
The critical rotation speed for the condensate to develop a central hole is 
\begin{equation}\label{Omegac}
\Omega _c = 1+\left(\frac{3k^2bG}{8\pi}\right)^{1/3}.
\end{equation}
For subcritical velocities, the behavior of the condensate is not qualitatively different from that of a harmonically trapped condensate, so we focus on velocities $\Omega \geq \Omega_{c}$. In our scaling (see (\ref{Omega})) this is equivalent to the requirement $\beta \geq (3bG)^{1/3}(8 \pi)^{-1/3}$. Then the inner and outer radii of the condensate $R_{\pm}$ are given by the relations
\begin{equation}\label{Rayons}
R_+ ^2 + R_- ^2 = 4\beta, \quad R_+ ^2 - R_- ^2 = \left(24bG\right) ^{1/3},
\end{equation}
the chemical potential is
\begin{equation}\label{mu}
\tfchem _b = \left(\frac{3bG}{8\pi}\right)^{2/3}-\beta^{2}
\end{equation}
and the energy is  
\begin{equation}\label{Energie}
\tfe_b =\frac{3}{5}\left(\frac{3bG}{8\pi}\right)^{2/3}-\beta ^2.
\end{equation}
In the case of a disc-shaped condensate, the order of magnitudes of the corresponding quantities is the same as above.\\
This information allows to check a posteriori the validity of our analysis. First, the LLL reduction should be valid if the energy $\lllee$ is much smaller than the spectral gap between the Lowest Landau Level and the first excited level. In our scaling, this gap is $\propto \ep^{-1} = k ^{-1/3}$, and thus from (\ref{upperbound}) and (\ref{Energie}) we deduce that the reduction is valid if 
\begin{equation}\label{reducLLL}
G \ll k ^{-1/2}, \: \beta \ll k^{-1/6}.
\end{equation}
One could probably show rigorously that the LLL reduction is valid under these conditions, working in the spirit of \cite{AB-2D}. Note that we have considered $\beta$ and $G$ as fixed in this paper so a rigorous proof would require to track the dependence on this parameters of our remainder terms. At least one condition will appear for sure if one wants to use our upper bound : we have used a (elementary) homogenization argument which requires  the number of vortices in the support of $\tfm_b$ to be large compared to $1$. Given that the vortices lie on a lattice of period $\propto \ep ^{1/2} = k  ^{1/6}$ and using (\ref{Rayons}) we arrive at the condition
\begin{equation}\label{reducHomogen}
 k\ll G.
\end{equation}
It is interesting to remark (tracking the dependence on $\beta$ and $G$ in the proof of Theorem \ref{theo:infivortex}) that under the above condition the number of vortices for the true minimizer should be infinite. We also note that, as we require both $k\ll G$ and $G\ll \frac{1}{\sqrt{k}}$, necessarily $k\ll 1$, which justifies our study of a small anharmonicity regime.\\

Finally, interesting questions arise about the vortex pattern displayed by a minimizer of the energy. If one believes that a minimizer is close to our trial state (\ref{deff}), then the property
\[
\Pie (\al \ft) \sim \al \ft
\]
that holds true in $L^{\infty}$ (amongst other topologies, see Lemma \ref{lem:ABN} below) indicates that the vortices should lie close to the sites of an hexagonal lattice in the region of significant matter density, that is the support of $\tfm_b$. Very little information however can be obtained about vortices lying in the region of low density.\\
These `invisible vortices' have nevertheless a contribution : as shown in \cite{AB,ABD}, they help create the average Thomas-Fermi density profile. Indeed, if the vortex pattern was regular in the whole complex plane, one would obtain a Gaussian density profile. More precisely it is to be expected that the regular vortex lattice will be distorted close to the edge of the condensate, resulting in invisible vortices that will help shape the condensate. An interesting question (see e.g. \cite{ABD,FJS,KTU}) is then, when $\Om$ is above $\Om_c$ and the condensate has an annular shape, what is the vortex distribution inside the hole ? In particular : do vortices gather in a central multi-quantized vortex ?\\
A nice feature of the LLL regime is that the condensate is completely determined by the location of its vortices, and one can thus minimize the energy numerically by varying the locations of the zeros of the wave-function. This gives a direct access to the optimal vortex pattern and allows to spot precisely the `invisible vortices'. A direct minimization of the Gross-Pitaevskii energy would not give such an information because one should a posteriori look for density dips in a region where the density is already very small.\\
We have adapted the numerical method of \cite{ABD} to the present setting and focused on the case where the condensate is annular. We refer to Section \ref{sec:numeric} for details. In particular we find invisible vortices both outside the outer edge of the condensate and in the central hole of low density, with multiply quantized vortices appearing at the center of the trap. We also find a good agreement between the optimal value of the energy numerically computed and the upper bound (\ref{upperbound}), thus giving another argument for the optimality of our construction.\\

The rest of the paper is organized as follows : in Section \ref{sec:proofs} we prove our analytical results, the proof of each theorem occupying a subsection. In Section \ref{sec:numeric} we describe our numerical method and present our results.

\section{Proofs}\label{sec:proofs}

\subsection{The Minimization Problem}\label{sousec:miniEL}

We prove Theorem \ref{theo:minimisation}. In this subsection, all parameters are considered as fixed.
 
\begin{proof}[Proof of Theorem \ref{theo:minimisation}]
We first note (see (\ref{lowerbound})) that the energy functional $\lllff$ is indeed bounded below in the minimization domain we have chosen (\ref{miniLLL}). \\
Our energy is defined only for $f\in \Fo ^2$, so we take the convention that $\lllff [f] = +\infty$ if $f\notin \F_{\ep} ^2$. Let $(f_n) _{n\in \N}$ be a minimizing sequence for (\ref{miniLLL}) and $u_n=f_n e^{-|z|^2/ \ep}$. We note that 
\[
\int_{\C} |z| ^2 |u_n| ^2 \leq \left(\int_{\C}  |u_n| ^2  \right)^{1/2} \left( \int_{\C} |z| ^4 |u_n| ^2 \right)^{1/2} = \left( \int_{\C} |z| ^4 |u_n| ^2 \right)^{1/2}
\]
because of the mass constraint. We deduce that  
\[
\lllff [f_n] \geq \frac{1}{4}  \int_{\C} |z| ^4 |u_n| ^2 dz - |\beta| \left( \int_{\C} |z| ^4 |u_n| ^2 dz \right)^{1/2} 
\]
and thus the sequence $(|z|^2 u_n)_n$ is bounded in $L^2$ whatever the sign of $\beta$. This implies that $(f_n)_n$ is bounded in $\Fo ^2$. In \cite{ABNmath} it is proved by using the Bargmann transform that $\Fo ^1$ is compactly embedded in $\Fo$. Similarly $\Fo ^2$ is compactly embedded \footnote{The Bargmann transforms maps $\Fos$ onto the  space $\left(1 -\dd_x ^2 + x ^2\right)^{-s}L ^2(\R)$.} in $\Fo ^1$. We thus have (possibly after extraction of a subsequence) the convergence of $f_n$ to some limit analytic function $f$, strongly in $\Fo^1$, strongly in $\Fo$ and weakly in $\Fo ^2$. We deduce that 
\[
\inf_n \int_{\C} \left( -\beta \vert z \vert ^2  + \frac{1}{4} \vert z \vert ^4 \right)\vert f_n \vert ^2 e^{-\vert z \vert^2 /\varepsilon} \geq \int_{\C} \left( -\beta \vert z \vert ^2  + \frac{1}{4} \vert z \vert ^4 \right)\vert f \vert ^2 e^{-\vert z \vert^2 /\varepsilon}. 
\] 
and
\[
\Vert f \Vert_{\Fo} = 1.
\]
On the other hand, denoting
\begin{equation}\label{normeAp}
\Vert f \Vert_{A^p _{\ep}} := \left( \frac{1}{\pi h}\int_{\C} |f(z)| ^p e^{-|z|^2/ \ep} dz \right) ^{1/p}
\end{equation}
it is proved in \cite[Theorem 4]{Car} (see also \cite[Section 2.2]{ABNmath}) that for any entire function $f$
\begin{equation}\label{hypercontract}
\Vert f (e^{-t} \cdot ) \Vert_{A^q _{\ep}} \leq  \Vert f  \Vert_{A^p _{\ep}},\mbox{ for any } 0<p<q \mbox{ such that } e^{-t} \leq \frac{p ^{1/2}}{q ^{1/2}}.
\end{equation}
where $f(e^{-t}\cdot)$ is the function mapping $z$ to $f(e^{-t}z)$. Taking $t = \frac{\log 2}{2}$, $p=2$ and $q=4$ we deduce that for any entire function $f$
\[
\left( \int _{\C} \vert f \vert ^4 e^{-2\vert z \vert^2 /\varepsilon}dz \right) ^{1/4} \leq C  \Vert f  \Vert_{\Fo}.
\]
Thus our sequence $(f_n)_n$ is also bounded in $A^4_{\ep}$ and we can assume (after a possible further extraction) that the convergence of $f_n$ to $f$ holds also in the weak topology of $A^4_{\ep}$. In particular, by convexity, there holds
\[
 \inf_n \int _{\C} \vert f_n \vert ^4 e^{-2\vert z \vert^2 /\varepsilon}dz \geq \int _{\C} \vert f \vert ^4 e^{-2\vert z \vert^2 /\varepsilon}dz 
\]
which concludes the proof that 
\[
\lllff [f] \leq \inf _n \lllff [f_n] 
\]
and thus the proof that $f$ minimizes $\lllff$ under the mass constraint.\\
We turn to the derivation of the Euler-Lagrange equation. The weak derivative of the functional $\lllff$ at $f$ along $g$ is given by
\begin{equation*}
D \F ^{\mathrm{LLL}}(f) \cdot g = \int_{\C} \left(\left( -\beta \vert z \vert ^2  + \frac{1}{4} \vert z \vert ^4 \right)\bar{f}g e^{-\vert z \vert^2 /\varepsilon} + \frac{G}{2} \vert f \vert ^2\bar{f}g e^{-2\vert z \vert^2 /\varepsilon}\right)dz.
\end{equation*}
Using an integration by parts and $\partial_{\bar{z}}f=\partial_{\bar{z}}g=0$ on the first term and the fact that $\Pi_{\varepsilon}(g)=g$ on the second term, we obtain (\ref{EEL1}).\\
Equation (\ref{EEL2}) is obtained from (\ref{EEL1}) exactly as in \cite[Proposition 3.2]{ABNmath} with some algebra on the non-linear term. To get (\ref{EEL3}) we use an integration by parts to show that $M_{\varepsilon}f=\Pi_{\varepsilon}(\vert z\vert^2 f)$. Then $M_{\varepsilon}^2f=\Pi_{\varepsilon}(\vert z\vert^2 \Pi_{\varepsilon}(\vert z\vert^2 f))$ and we get the result.\\

\end{proof}

\subsection{Infinite Number of Zeros}\label{sousec:zeros}

We now prove that any minimizer of (\ref{miniLLL}) has an infinite number of zeros. The argument is by contradiction and in two steps, using the two equivalent forms of the Euler-Lagrange equation (\ref{EEL1}) and (\ref{EEL2}).

\begin{proof}[Proof of Theorem \ref{theo:infivortex}]\mbox{}\\

\emph{Step 1.} Suppose $f$ has a finite number of zeros. Then one may write $f(z)=P(z)e^{\varphi(z)}$ where $P$ is a polynomial and $\varphi$ is a holomorphic function. Now $f\in \Fo$ and the condition $\int_{\C} \vert f\vert^2 e^{-\vert z \vert^2/\varepsilon}dz < \infty$ implies that $Re(\varphi(z))\leq \vert z\vert^2 /(2 \ep)$. It is well-known (see \cite{Boa} for example) that a holomorphic function can satisfy this condition only if it is a polynomial of degree less than 2. Therefore we know that 
\begin{equation}\label{f=Pephi}
 f(z)=P(z)e^{\alpha_1 z +\alpha_2 z^2}
\end{equation}
and the integrability condition on $f$ implies $\alpha_2 \leq 1/(2 \ep)$. Injecting (\ref{f=Pephi}) in (\ref{EEL1}) and comparing the exponential growth of the different terms of (\ref{EEL1}) as in \cite{ABNmath} yields $\alpha_1=\alpha_2=0$. So, if $f$ has a finite number of zeros, it is a polynomial.\\

\emph{Step 2.} Now, suppose $f$ is a polynomial of degree $n$ 
\[
f(z) = \sum_{k=0} ^n a_k z ^k
\]
and inject this in (\ref{EEL2}). The term 
\[
-\beta M_{\varepsilon} f+\frac{1}{4}\left( M_{\varepsilon}^2 f + \varepsilon M_{\varepsilon} f\right)
\]
 is a polynomial of degree $n$, by the very definition (\ref{defiMep}) of $M_{\ep}$. Therefore (\ref{EEL2}) implies that the term 
\[
\frac{G}{2}\bar{f}(\varepsilon \partial_{z})[f^2(z/2)] = \sum_{k=0} ^n (\ep \dd_z )^k [f^2 (z/2)]
\] 
is also of degree $n$. But $(\ep \partial _z )^k[f^2(z/2)]$ is of degree $2n-k$, so that it must be that $a_k=0$ for any $k<n$. Then $f$ is of the form 
\[
f(z)=a_n z^n,
\]
with 
\[
|a_n|^2 = \left( \pi \ep ^{n+1} n! \right) ^{-1}
\]
because of the normalization of $f$ in $\Fo$. Injecting this a last time in (\ref{EEL2}) yields
\[
-\beta (n+1)\ep + \frac{G}{2} \frac{(2 n)!}{\pi \ep 2^{2n +1}(n!)^2}-\mu + \frac{1}{4}(2\ep ^2 +n\ep ^2 + n(n-1)\ep ^2)= 0.
\]
Using the improved Stirling formula \cite{Rob} we obtain a condition on $n$:
\begin{equation}\label{n}
\mu + \beta \ep \geq -\beta n \ep + \frac{Ge^{-1/12}}{2\pi \ep \sqrt{n}} \\ 
+\frac{\ep^2}{2}+\frac{n^2 \ep^2}{4}+\frac{3n\ep^2}{4}.
\end{equation}
We now need to bound the chemical potential $\mu$ : taking the $\F_{\ep}$-scalar product of each side of (\ref{EEL1}) with $f$ yields
\begin{equation}\label{bornemu}
 \mu \leq 2\lllff[f] +\beta^2 - \frac{\beta \ep}{2}+\frac{\ep^2}{16}.
\end{equation}
Here we used the fact that the spectrum of $-\beta M_{\varepsilon} +\frac{1}{4}\left( M_{\varepsilon}^2  + \varepsilon M_{\varepsilon} \right)$ is bounded below by $-\beta^2 + \frac{\beta \ep}{2}-\frac{\ep^2}{16}$. This follows by noting that the spectrum of $M_{\ep}$ is constituted by the eigenvalues $(n+1)h$ with (non-normalized) eigenvectors $z^n$. We use the upper bound on $\lllff[f]$ of Theorem \ref{theo:pointscri} (see Subsection \ref{sousec:pointscri} for the proof) to deduce
\begin{equation}\label{n2}
\frac{6}{5}\left( \frac{3bG}{8\pi}\right)^{2/3} - \beta^2  - \frac{\beta \ep}{2}-\frac{7\ep^2}{16}
\geq -\beta n \ep + \frac{Ge^{-1/12}}{2\pi \ep \sqrt{n}}+\frac{n^2 \ep^2}{4}+\frac{3n\ep^2}{4}.
\end{equation}
Minimizing the right-hand side of (\ref{n2}) with respect to $n$ (taken as a continuous variable as it should be very large when $\ep$ is small) for fixed $\ep$ yields
\begin{equation}\label{eq:n3}
n \sim \left( \frac{Ge^{-1/12}}{2\pi} \right) ^{2/5}\ep ^{-6/5}
\end{equation}
and we deduce
\[
G ^{2/3} \geq C \left( G^{4/5} \ep ^{-2/5} - \beta G ^{2/5} \ep ^{-1/5} + \beta ^2 \right)
\]
which is a contradiction if $\ep$ is small enough.
\end{proof}

\subsection{Construction of Critical Points}\label{sousec:pointscri}

The proof of Theorem \ref{theo:pointscri} requires the following result taken from \cite{ABNmath}, that we quote for convenience. Here we abuse notations by taking the $\Fos$ norm of non-holomorphic functions.  

\begin{lem}[\textbf{ \cite{ABNmath} Estimates for $\fat$}]\label{lem:ABN}\mbox{}\\
Let $\gamma \in C^{0,1/2} (\C,\C)$ have compact support in $B_{R_0}$. If $g\in \Fos $ for some $s$, then $\Pie (\gamma g) \in \Fo ^{s'}$ for any $s' \in \R$ and
\begin{equation}\label{estim fat fock}
\left\Vert  \Pie (\gamma g) -\gamma g  \right\Vert_{\Fo ^{s'}} \leq  C_{s,s',R_0} \Vert \gamma \Vert_{C^{0,1/2}}\Vert g \Vert_{\Fo ^s} \ep^{1/4}.
\end{equation}
If $v = e^{-|z|^2 / \ep} g \in L^p (B_{R_0})$ for $1\leq p \leq +\infty$ then
\begin{equation}\label{estim fat Lp}
\left\Vert \left( \Pie (\gamma g) -\gamma g \right) e^{-|z|^2/\ep} \right\Vert_{L ^p} \leq  C_{p,R_0} \Vert \gamma \Vert_{C^{0,1/2}} \Vert v \Vert_{L^p (B_{R_0})}.
\end{equation}
\end{lem}

We now provide the

\begin{proof}[Proof of Theorem \ref{theo:pointscri}]

Let $\al \in C^{0,1/2} (\C,\C)$ have a fixed compact support and satisfy
\begin{equation}\label{masse alpha}
\int_{\C} |\al| ^2=1.
\end{equation}
We first claim that for any $\nu \in \R$
\begin{multline}\label{quasiequa1}
(-\beta+\frac{1}{4}\varepsilon)\Pi_{\varepsilon}(\vert z \vert^2 f_{\alpha,\tau})+\frac{1}{4}\Pi_{\varepsilon}\left( \vert z\vert^2 \Pi_{\varepsilon}(\vert z\vert^2 f_{\alpha,\tau})\right) + G\Pi_{\varepsilon}(e^{-\vert z \vert^2 /\varepsilon}\vert f_{\alpha,\tau}\vert^2 f_{\alpha,\tau})- \nu f_{\alpha,\tau} \\
= \Vert \Pie (\al \ft) \Vert_{\Fo} ^{-1} \Pie \left( \left( -\beta |z| ^2 + \frac{1}{4} |z|^4 + G b |\al| ^2 -\nu \right) \al \ft \right) +R_{\ep},
\end{multline}
with $\Vert R_{\ep} \Vert_{\Fo}\leq C \ep^{1/4}$. The proof consists in a repeated use of Lemma \ref{lem:ABN}.\\
Let $p$ be a polynomial. We estimate, using Cauchy-Schwarz and (\ref{estim fat fock}) applied with $g=\ft$ and $\gamma = \al$,
\[
\left| \int_{\C} \left(  |z| ^2 \al \ft - |z|^2  \Pie (\al \ft) \right)\bar{p} e^{-|z|^2/ \ep}dz \right| \leq C \Vert p \Vert_{\Fo} \Vert \Pie(\al \ft) - \al \ft \Vert_{\Fo ^2 }\leq C \ep ^{1/4} \Vert p \Vert_{\Fo}.
\]
Thus, by density of polynomial in $\Fo$
\begin{equation}\label{preuve points cri 1}
\left\Vert \Pi_{\varepsilon}(\vert z \vert^2 \Pie (\al \ft)) - \Pi_{\varepsilon}(\vert z \vert^2 \al \ft ) \right\Vert_{\Fo} \leq C\ep^{1/4}.
\end{equation}
Similarly, applying (\ref{estim fat fock}) with $\gamma = \al |z| ^2$
\[
\left| \int_{\C} \left(  |z| ^4 \al \ft - |z|^2  \Pie (\al |z| ^2  \ft) \right)\bar{p} e^{-|z|^2/ \ep}dz \right| \leq C \Vert p \Vert_{\Fo} \Vert \Pie(\al |z|^2 \ft) - \al |z|^2 \ft \Vert_{\Fo ^2 }\leq C \ep ^{1/4}  \Vert p \Vert_{\Fo}
\]
and
\begin{eqnarray*}
\left| \int_{\C} \left(  |z| ^2 \Pie (\al |z| ^2 \ft) - |z|^2  \Pie ( |z| ^2  \Pie (\al \ft)) \right)\bar{p} e^{-|z|^2/ \ep}dz \right| &\leq&   \Vert p \Vert_{\Fo} \left \Vert \Pie \left( \al |z|^2 \ft - |z|^2 \Pie (\al \ft)\right) \right\Vert_{\Fo ^2 } 
\\ &\leq&  \Vert p \Vert_{\Fo} \left\Vert \al  \ft -  \Pie (\al \ft) \right\Vert_{\Fo ^4 }\leq C \ep ^{1/4}  \Vert p \Vert_{\Fo}
\end{eqnarray*}
which implies
\begin{equation}\label{preuve points cri 2}
\left\Vert \Pi_{\varepsilon}\left( \vert z\vert^2 \Pi_{\varepsilon}(\vert z\vert^2 \Pie (\al \ft) )\right)  - \Pi_{\varepsilon}(\vert z \vert^4 \al \ft) \right\Vert_{\Fo} \leq C\ep^{1/4}.
\end{equation}
We turn to the non linear term in the equation. Using (\ref{fabri}) we have for any polynomial $p$
\begin{eqnarray}
 b \Xint - |\ut| ^2 \int_{\C} |\al| ^2 \al \Pie \left( e^{-|z|^2/\ep }  |\ft| ^2 \ft  \right) \bar{p} e ^{-|z|^2/\ep }dz &=& \int_{\C} |\al| ^2 \al \ft \bar{p} e ^{-|z|^2/\ep }dz \nonumber \\
&=& \int_{\C}  |\al| ^2 \al e^{-|z|^2/\ep }  |\ft| ^2 \ft \bar{p} e ^{-|z|^2/\ep } dz \nonumber \\
&+& O (\ep ^{1/4}) \label{use eq ft}.
\end{eqnarray}
The second line is a consequence of (\ref{estim fat fock}) applied with $\gamma = |\al| ^2 \al$ and $p=g$. On the other hand, using Lemma \ref{lem:ABN} again
\begin{multline*}
\left| \int_{\C}  \left( e^{-\vert z \vert^2 /\varepsilon}\vert \Pie (\al \ft) \vert^2 \Pie (\al \ft) - |\al| ^2 \al e^{-|z|^2/\ep }  |\ft| ^2 \ft \right) \bar{p} e ^{-|z|^2/\ep }dz \right| \\
\leq  \Vert p \Vert_{\Fo}  \Vert \Pie (\al \ft) \Vert_{\Fo} \Vert |\al| ^2 |\ut| ^2  - \left|\Pie (\al \ft) \right| ^2 e^{-|z|^2/\ep} \Vert_{L^{\infty}} 
\\ + \Vert p \Vert_{\Fo}  \Vert \al  \Vert_{L^{\infty}} ^2 \Vert \ut  \Vert_{L^{\infty}} ^2 \Vert \al \ft - \Pie (\al \ft) \Vert_{\Fo}  \leq   \ep ^{1/4} \Vert p \Vert_{\Fo} 
\end{multline*}
and thus
\begin{equation}\label{preuve points cri 3}
\left\Vert \Pi_{\varepsilon}(e^{-\vert z \vert^2 /\varepsilon}\vert \Pie (\al \ft) \vert^2 \Pie (\al \ft) - b \left(\Xint - |\ut| ^2\right) |\al| ^2  \al \ft \right\Vert_{\Fo} \leq C \ep ^{1/4}.
\end{equation}
Recall that 
\[
\fat = \frac{\Pie (\al \ft)}{\Vert \Pie (\al \ft) \Vert_{\Fo}}.
\]
We claim that
\begin{equation}\label{norme fat}
\Vert \Pie (\al \ft) \Vert_{\Fo} ^2 = \Xint - |\ut| ^2 \left( 1 + O (\ep ^{1/4})\right).
\end{equation}
Indeed, from Lemma \ref{lem:ABN} we have
\[
\Vert \Pie (\al \ft) \Vert_{\Fo} ^2  = \int_{\C} |\al| ^2 |\ut| ^2 + O(\ep ^{1/4}).
\]
On the other hand $| \ut| $ is periodic over a lattice of period proportional to $\ep ^{1/2}$, thus for any $\gamma \in C ^{\infty}(\C)$
\[
\left| \int_{\C} |\gamma| ^2 |\ut| ^2 - \Xint - |\ut| ^2 \int_{\C} |\gamma| ^2 \right| \leq C \Vert \gamma \Vert_{C^0 (\C)}
\]
and
\[
\left| \int_{\C} |\gamma| ^2 |\ut| ^2 - \Xint - |\ut| ^2 \int_{\C} |\gamma| ^2 \right| \leq C \ep ^{1/2} \Vert \gamma \Vert_{C^1 (\C)}.
\]
Using $\al \in C ^{0,1/2}$ and interpolating between the two previous estimates we obtain
\begin{equation}\label{homogene}
\int_{\C} |\al| ^2 |\ut| ^2 = \Xint - |\ut| ^2 \int_{\C} |\al| ^2 + O(\ep ^{1/4})
\end{equation}
and thus (\ref{norme fat}), recalling (\ref{masse alpha}). Gathering (\ref{preuve points cri 1}),(\ref{preuve points cri 2}),(\ref{preuve points cri 3}) and (\ref{norme fat}) we deduce  that (\ref{quasiequa1}) holds.\\
There only remains to choose $\al$ satisfying (\ref{tfprofile}) and $\nu = \tfchem_b$ to have
\[
\left( -\beta |z| ^2 + \frac{1}{4} |z|^4 + G b |\al| ^2 -\nu \right)\al = 0
\] 
and thus (\ref{resultEq}).\\
The proof of the energy result (\ref{resultEn}) uses the same tools. By definition
\[
\lllff [\fat] = \Vert \Pie (\al \ft) \Vert_{\Fo} ^ {-2} \int_{\C} \left( -\beta |z| ^2 + \frac{1}{4} |z| ^ 4 \right)  \left|\Pie (\al \ft) \right| ^2 e^{-|z|^2 / \ep} +  \Vert \Pie (\al \ft) \Vert_{\Fo} ^ {-4} \frac{G}{2} \int_{\C} \left|\Pie (\al \ft) \right| ^4 e^{- 2 |z|^2 / \ep}
\]
working as above, with the $L^p$ estimate (\ref{estim fat Lp}) instead of (\ref{estim fat fock}) we deduce
\begin{equation}\label{energie1}
\lllff [\fat] = \Vert \Pie (\al \ft) \Vert_{\Fo} ^ {-2} \int_{\C} \left( -\beta |z| ^2 + \frac{1}{4} |z| ^ 4 \right)  \left|\al \ut \right| ^2 +  \Vert \Pie (\al \ft) \Vert_{\Fo} ^ {-4} \frac{ G}{2} \int_{\C}\left|\al \ut \right| ^4 + O(\ep ^{1/4}).
\end{equation}
Next, using that $|\ut|$ is periodic over a lattice of period proportional to $\ep ^{1/2}$ as above (see the proof of (\ref{homogene})), we obtain
\begin{eqnarray}
\lllff [\fat] &=& \Vert \Pie (\al \ft) \Vert_{\Fo} ^ {-2} \Xint - |\ut|^2 \int_{\C} \left( -\beta |z|^2 + \frac{1}{4} |z|^4 \right)  \left|\al \right| ^2 \nonumber
\\&+&  \Vert \Pie (\al \ft) \Vert_{\Fo} ^ {-4} \Xint - |\ut|^ 4 \frac{G}{2} \int_{\C }\left|\al \right| ^4 + O(\ep ^{1/4}).
\end{eqnarray}
There only remains to recall (\ref{btau}) and (\ref{norme fat}) to conclude that (\ref{resultEn}) holds.

% Passing to the second line is a straightforward computation (see \cite{ABD,FJS} for details). Note that $\tfm_b$ defined in (\ref{tfprofile}) is the unique positive minimizer of
%\[
%\tff_b [\rho] = \int_{\C} \left( \left( -\beta \vert z \vert ^2  + \frac{1}{4} \vert z \vert ^4 \right) | \tfm_b |  + \frac{bG}{2} \vert \tfm_b \vert ^2 \right)dz
%\]
%under the constraint
%\[
%\int_{\R ^2} \rho = 1.
%\]

\end{proof}

\section{Numerical Simulations}\label{sec:numeric}

\subsection{Method}\label{sousec:numapproach}

We want to find a numerical approximation of the minimizer $\phi$ of $\E ^{LLL}$ in the LLL. We write 
\begin{equation}\label{eq:LLLnum}
 \phi(z)=P(z)e^{-\Omega \vert z \vert ^2 /2}
\end{equation}
with $P$ a holomorphic function. Since polynomials are dense in $\Fo$, it is reasonable to fix an integer $n$ and to restrict the analysis to functions $\phi$ where $P$ is a polynomial of degree less than $n$ (see \cite[Section 6]{ABNmath} for mathematical results on the validity of this approximation). We write our trial functions as in \cite{ABD}:
\begin{equation}\label{eq:ftest}
\phi(z)=A \prod ^n _{j=1} (z-z_j)e^{-\Omega \vert z \vert ^2 /2},
\end{equation}
where 
\[
A=\Vert \phi \Vert _{L^2}^{-1/2} 
\]
is the normalization factor. We will numerically vary the locations $z_j$ of the vortices. An alternative method (used for example in \cite{CKR} for a harmonically trapped condensate) would be to take 
\begin{equation*}
\phi(z)=A\left( \sum ^n _{j=0} b_j z^j \right) e^{-\Omega \vert z \vert ^2 /2}
\end{equation*}
and vary the coefficients $b_j$. The interest of our approach is to give a direct access to the exact repartition of vortices, whereas the alternative method would require to compute the roots of a polynomial of degree $n$, which is a delicate task for large $n$. In particular, varying the coefficients could probably not give the precise locations of invisible vortices located in regions of low matter density. Note also that there exist some good methods to directly minimize the GP energy (see \cite{DH,DK} for recent developments), but they probably could not give access to the locations of invisible vortices neither. Indeed, one would have to spot zeros or phase singularities in a region where the density is very small.\\  
The point here is that these vortices, although invisible, have a crucial influence on both the energy and the density profile of the condensate. One can be skeptical about the existence of density dips in a region where there is so to say no matter at all. But a vortex is also carrying a phase circulation (or superfluid current) and thus one can interpret the presence of invisible vortices as the existence of currents in the condensate that are equivalent to those generated by vortices in the low density region. Taking invisible vortex into account with a trial function of the form (\ref{eq:ftest}) is thus a way to evaluate more precisely the effect of superfluid currents in the condensate.\\ 
We used a conjugate gradient method with a Goldstein and Price line-search. All the functions whose integrals have to be computed are of the form Polynomial $\times$ Gaussian. We thus used the Gauss-Hermite method, and we took enough Gauss points for the computations to be exact. This results in quite expensive calculations, but we have been able to numerically construct condensates with up to $\sim 120$ vortices.

\subsection{Results}\label{sousec:numresult}

\begin{figure}[h]
\begin{centering}
\includegraphics[width=120mm]{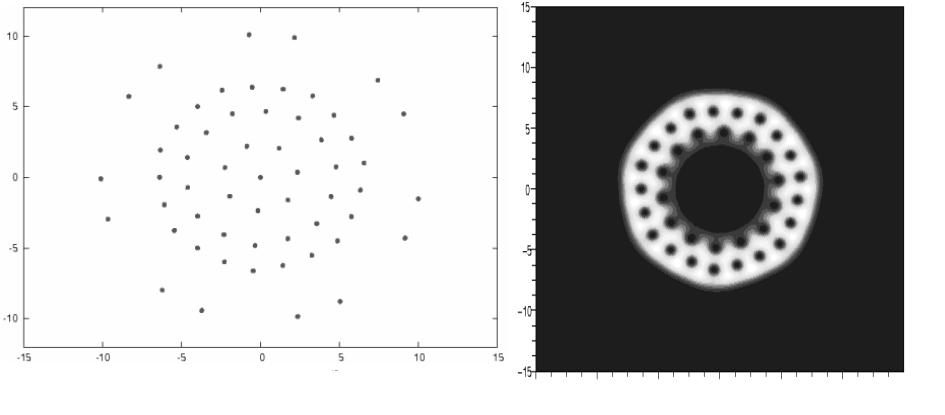}
\caption{Vortex structure and atomic density for $G=3, k=10^{-4}, \beta=1 \: (\Omega-1 = 2.2 \, 10^{-3})$. There are 67 vortices in total, the central vortex is constituted of 11 single vortices. \label{fig1}}
\end{centering}
\end{figure}

We show in Fig.\ref{fig1} and Fig.\ref{fig2} typical examples of configurations we numerically computed. The qualitative features of the vortex patterns and atomic densities confirm our theoretical results and are in good agreement with existing theoretical and numerical studies \cite{FJS,KB}. Note however that the numerics become quite intricate for large number of vortices, which accounts for the relative lack of symmetry of Fig.\ref{fig2}.\\
\begin{figure}[h]
\begin{centering}
\includegraphics[width=120mm]{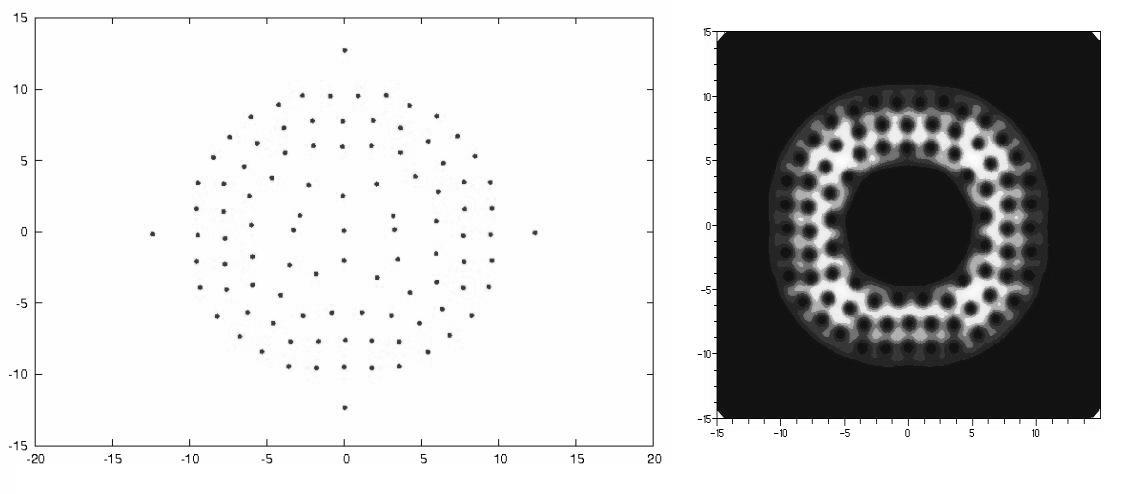}
\caption{Vortex structure and atomic density for $G=3, k=10^{-5}, \beta=1 \: (\Omega-1 = 4.6 \, 10^{-4})$. There are 119 vortices in total, the central vortex is constituted of 20 single vortices. \label{fig2}}
\end{centering}
\end{figure}
\begin{figure}[h]
\begin{centering}
\includegraphics[width=120mm]{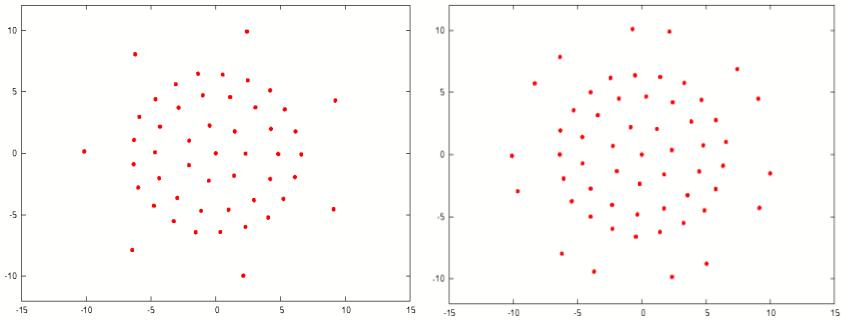}
\caption{Two example of vortex configurations for $G=3, k=10^{-4}, \beta=1 \: (\Omega-1 = 2.2 \, 10^{-3})$, respectively with $n=60$ and $n=67$ vortices. \label{fig3}}
\end{centering}
\end{figure}
As was expected, the condensate develops a central hole and visible vortices are regularly distributed and densely packed in the annular region of significant atomic density. The hexagonal lattice is not clearly observed however, which seems to be due to the competition between the annular geometry of the condensate and the repulsion between vortices. The former tends to force vortices to form concentric circles while the latter favors the hexagonal lattice. This effect is apparent in several other simulations of rotating annular BECs, see \cite{AM,FJS} for 2D simulations and \cite{AD,D} for 3D ones. To observe clearly the hexagonal lattice one would need the annulus to be thicker so that it could contain more vortices. This would make the numerics all the more challenging.\\   

Our computations also show a distortion of the vortex pattern near the external radius of the condensate, resulting in invisible vortices in the exterior region of low atomic density as is the case for harmonically trapped condensates \cite{ABD}. Some vortices also lie in the central hole as theoretically predicted and we can get information on their precise locations: we observe a distortion of the regular lattice near the inner radius of the condensate, resulting in isolated singly-quantized vortices encircling a central multiply quantized vortex. We computed configurations for which this central vortex has up to $20$ units of circulation while there is a total of $32$ units of circulation in the entire hole and $83$ visible vortices. The number of vortices (both visible and invisible) increases with increasing $\Omega$ or $\beta$, but since our scaling does not allow to explore a large domain of $\Omega$ when $G$ is fixed, we mainly varied $k$. The total vorticity of the system increases with decreasing $k$.\\

All vortices do not have the same contribution to the energy: as $n$ increases, the vortex pattern in the annular region of significant atomic density remains the same up to possible rotations, with the additional vortices first gathering in the central vortex, then constituting the distorted lattice near the inner boundary of the condensate and finally occupying the distorted sites beyond the external radius. With increasing $n$ (see Fig.\ref{fig4}), the energy reaches a first plateau when the central figure is formed, constituted of the visible vortices and the vortices in the central hole. A second plateau is reached when enough distorted sites beyond the external radius are occupied. For example, for the parameters corresponding to Fig.\ref{fig1} and Fig.\ref{fig4} ($k=10^{-4},\: G=3, \: \beta=1$), the energy varies by $\sim \pm 10^{-6}$ in relative value when $n$ is increased from $60$ to $67$ (Fig.\ref{fig3}) and the atomic density does not vary significantly. We find a good agreement between our numerical results and the upper bound to the energy computed in Theorem \ref{theo:pointscri}. The energy of our analytical trial function typically differs by $\sim 10^{-3}$ to $\sim 10^{-2}$ in relative value from the energy numerically computed. The agreement between the radii theoretically predicted (\ref{Rayons}) and extracted from the numerical simulations is typically a little worse but still satisfactory (relative errors ranging from $10^{-2}$ to $10^{-1}$). Again, obtaining a better agreement would probably require to consider a thicker annulus containing more vortices. This would result in more difficult simulations.

\begin{figure}[ht]
\begin{centering}
\includegraphics[width=120mm]{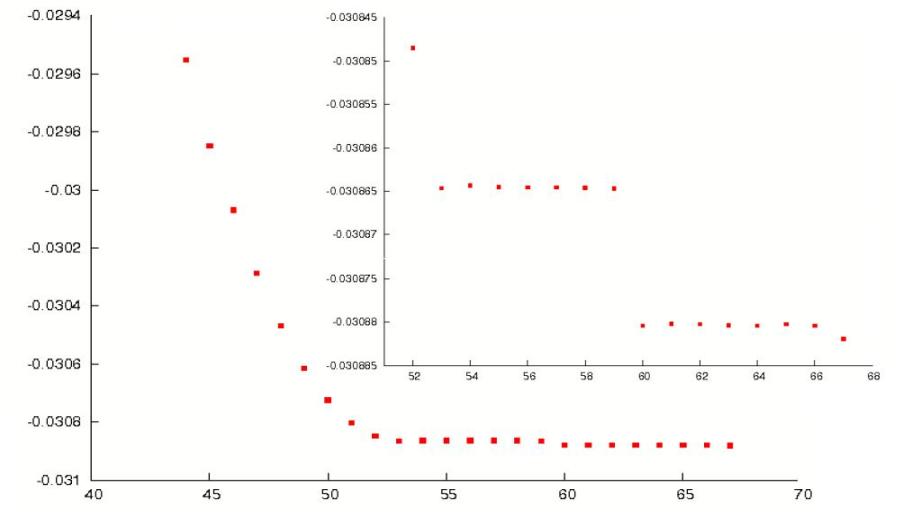}
\caption{Minimum energy as a function of the number of vortices in the trial wave function ($G=3, k=10^{-4}, \beta=1$). \label{fig4}}
\end{centering}
\end{figure}

\end{document}